\documentclass[12pt, draftclsnofoot, onecolumn]{IEEEtran}

\usepackage{amssymb}
\usepackage{amsmath}
\usepackage{amsmath,bm}
\usepackage{amsthm}
\usepackage{graphicx}
\usepackage{subfigure}
\usepackage{cite}
\usepackage{enumerate}
\usepackage{url}
\usepackage{bm}
\usepackage[ruled]{algorithm2e}
\usepackage{setspace}
\usepackage{color, soul}
\usepackage{epstopdf}
\usepackage[noend]{algpseudocode}
\usepackage{float}

\usepackage{booktabs}
\usepackage{threeparttable}
\usepackage{longtable}
\usepackage{rotating}
\usepackage{multirow}
\usepackage{array}
\usepackage{stfloats}

\newtheorem{prop}{Proposition}

\newtheorem{lemma}{Lemma}

\begin{document}

\title{\LARGE{Beyond Cell-free MIMO: Energy Efficient Reconfigurable Intelligent Surface Aided Cell-free MIMO Communications}}
\author{
\IEEEauthorblockN{
Yutong Zhang,
Boya Di, \IEEEmembership{Member, IEEE},
Hongliang Zhang, \IEEEmembership{Member, IEEE},
Jinlong Lin,
Yonghui Li, \IEEEmembership{Fellow, IEEE},
Lingyang Song, \IEEEmembership{Fellow, IEEE}}

\thanks{Yutong Zhang and Lingyang Song are with the Department of Electronics, Peking University,  Beijing, China (email: yutongzhang@pku.edu.cn; lingyang.song@pku.edu.cn).}
\thanks{Boya Di is with the Department of Electronics, Peking University, Beijing, China, and also with Department of Computing, Imperial College London, London, UK(email: diboya92@gmail.com).}
\thanks{Hongliang Zhang is with the Department of Electronics, Peking University, Beijing, China, and also with the Department of Electrical Engineering, Princeton University, Princeton, NJ, USA (email: hongliang.zhang92@gmail.com).}
\thanks{Jinlong Lin is with School of Software and Microelectronics, Peking University, Beijing, China (email: linjl@ss.pku.edu.cn).}
\thanks{Yonghui Li is with School of Electrical and Information Engineering, The University of Sydney, Sydney, Australia (email: yonghui.li@sydney.edu.au).}
}
\maketitle
\vspace{-4em}
\begin{abstract}
Cell-free systems can effectively eliminate the inter-cell interference by enabling multiple base stations (BSs) to cooperatively serve users without cell boundaries at the expense of high costs of hardware and power sources due to the large-scale deployment of BSs.
To tackle this issue, the low-cost reconfigurable intelligent surface (RIS) can serve as a promising technique to improve the energy efficiency of cell-free systems.
In this paper, we consider an RIS aided cell-free MIMO system where multiple RISs are deployed around BSs and users to create favorable propagation conditions via reconfigurable reflections in a low-cost way, thereby enhancing cell-free MIMO communications.
To maximize the energy efficiency, a hybrid beamforming (HBF) scheme consisting of the digital beamforming at BSs and the RIS-based analog beamforming is proposed.
The energy efficiency maximization problem is formulated and an iterative algorithm is designed to solve this problem.
The impact of the transmit power, the number of RIS, and the RIS size on energy efficiency are investigated.
Both theoretical analysis and simulation results reveal that the optimal energy efficiency depends on the numbers of RISs and the RIS size.
Numerical evaluations also show that the proposed system can achieve a higher energy efficiency than conventional ones.

\end{abstract}

\begin{keywords}
\small
\centering
{Reconfigurable intelligent surface, Cell-free, MIMO, Hybrid beamforming, Energy efficiency}
\end{keywords}

\section{Introduction}%

During the past few decades, multiple input and multiple output~(MIMO) has drawn a great attention~\cite{SciChina,ChinaCom}.
However, the performance of multi-cell MIMO systems is typically limited by inter-cell interference due to the \emph{cell-centric} implementation.
To address this issue, the concept of cell-free networks has been proposed as a \emph{user-centric} paradigm to enable multiple randomly-distributed base stations~(BSs) without cell boundaries to coordinate with each other to serve all users in the network simultaneously~\cite{cell-free,cell-free1}.
Nevertheless, the traditional cell-free system requires a large-scale deployment of BSs, leading to an unsatisfying energy efficiency performance due to the high costs of both hardware and power sources.

To tackle this issue, the reconfigurable intelligent surface~(RIS) has emerged as a potential cost-efficient technique by creating favorable propagation conditions from BSs and users~\mbox{\cite{MRenzo-survey, cuitiejun-survey}}.
Benefited from a large number of RIS elements whose phase shifts are controlled by simple programmable PIN diodes, RISs can reflect signals and generate directional beams from BSs to users~\cite{song-survey}.
Unlike large-scale phased array antennas in the cell-free system enabled by phase shifters with inevitable power consumption, RIS requires no extra hardware implementation, such as complex digital phase shift circuits, thus greatly saving the energy consumption and complexity for signal processing~\cite{lowcost}.
Hence, compared with the conventional cell-free systems, a lower level of power consumption is required to achieve the same quality-of-services (QoS).
In other words, RIS provides a new dimension for cell-free systems to enhance the energy efficiency of the cell-free systems.

In the literature, several initial works have considered the sum rate maximization of RIS-aided cell-free systems~\cite{cellfree-poor} and the multi-cell MIMO systems assisted by one RIS~\cite{Lajos}.
Energy efficiency is only evaluated in the single-RIS aided single-cell systems~\cite{EE-ICC,EE-Lajos,EE-DRL,Chongwen Huang}.
In more detail, authors in~\cite{cellfree-poor} proposed a distributed design framework for cooperative beamforming to maximize the weighted sum rate of a cell-free system with multiple RISs.
In~\cite{Lajos}, an RIS was utilized to enhance the cell-edge performance in multi-cell MIMO communication systems.
Both the active precoding at BSs and the phase shifts at the RIS were jointly optimized to maximize the weighted sum rate of all users.
A more general multi-RIS scenario was also investigated.
In~\cite{EE-ICC}, the energy efficiency of the network was maximized by controlling the phase shifts of a single RIS and beamforming of the BS.
Authors in~\cite{Chongwen Huang} studied a multi-user multiple input single output~(MISO) system with a single RIS and developed a joint power allocation and phase shift scheme design to maximize the system energy efficiency.

However, the above works have not fully exploited the potential of multiple coordinated RISs which provide a cost-efficient solution to further improve the energy efficiency of the conventional cell-free systems.
To fill this gap in the literature, in this paper we consider an \emph{RIS aided cell-free MIMO} system where multiple BSs and RISs are coordinated to serve various users.
Benefited from the programmable characteristic of RIS elements which mold the wavefronts into desired shapes, the propagation environment can be reconfigurable in a low-cost way, thereby enhancing the cell-free MIMO communications.
To achieve favorable propagation, it is vitally important to determine the phase shifts of all RIS elements, each of which can be viewed as an antenna, inherently capable of realizing analog beamforming.
Since the RISs do not have any digital processing capability, we consider an HBF scheme~\cite{weiyu} where the digital beamforming is performed at BSs, and the RIS-based analog beamforming is conducted at RISs.
In such an RIS-aided cell-free system, the spatial resource can be better utilized benefited from the multipath effects, i.e., the direct and RIS-reflected paths, for each transmitted signal to improve the energy efficiency performance.

To fully explore how RIS as a cost-efficient technique influence the energy efficiency of the cell-free system, we aim to optimize the digital beamformer at BSs and phase shifts of RISs.
Challenges have arisen in such a system.
\emph{First}, unlike conventional cell-free systems, the propagation environment is more complicated due to the extra reflected links via RISs.
It is hard to obtain the optimal solution to the joint digital beamforming at BSs and phase shift configuration of RISs.
\emph{Second}, compared to single-RIS systems, the dispersed distribution of multiple RISs brings a new degree of freedom for energy efficiency maximization.
It remains to be explored how the number of RISs as well as the RIS size influences the performance of the multi-RIS aided cell-free system.
\renewcommand\arraystretch{1.4}
\begin{table}[t]
\small
    \centering
    \caption{Major Notations}
    \begin{tabular}{c|c}
        \hline
        \hline
        \textbf{Notation} & \textbf{Description}\\
        \hline
        \hline
        $N$ &   Number of BSs\\
        \hline
        $N_a$   &   Number of antennas of each BS\\
        \hline
        $K$ &   Number of users\\
        \hline
        $M$ &   Number of RISs\\
        \hline
        $L$ &   Size of each RIS\\
        \hline
        $b$ &   Quantization bits for discrete phase shifts\\
        \hline
        $\theta_{m,l}$  &   The phase shift of the $l$-th RIS element of the $m$-th RIS\\
        \hline
        $q_{m,l}$   &   The frequency response of the $l$-th RIS element of the $m$-th RIS\\
        \hline
        $\mathbf{H}$    &   Channel matrix\\
        \hline
        $\mathbf{H}_D$  &   Direct channel between BSs and users\\
        \hline
        $\mathbf{H}_{RU}^H$ &   Reflected channel between RISs and users\\
        \hline
        $\mathbf{H}_{BR}$ &    Reflected channel between BSs and RISs\\
        \hline
        $\mathbf{Q}$    &   Phase shift matrix\\
        \hline
        $\mathbf{V}_D$  &   Digital beamformer\\
        \hline
        $R_k$   &   Data rate of user $k$\\
        \hline
        $R$ &   Sum rate\\
        \hline
        $P_B^{(n)}$ &   Hardware static power consumption of BS $n$\\
        \hline
        $P_R^{(m,b)}$   &   Hardware static power consumption of each element of RIS $m$ with 1-bit phase shifts\\
        \hline
        $P_U^{(k)}$ &   Hardware static power consumption of user  $k$\\
        \hline
        $\mathcal{P}$   &   Total power consumption\\
        \hline
        $p_t^{(n)}$ &   Transmit power at BS $n$\\
        \hline
        $B$ & Transmission bandwidth\\
        \hline
        $P_T^{(n)}$ & Transmit power budget available at each BS $n$\\
        \hline
        \hline
    \end{tabular}   \label{tab:notation}
\end{table}

By addressing the above challenges, we contribute to the state-of-the-art research in the following ways.
\begin{enumerate}
\item  	We propose an RIS aided cell-free MIMO system where various users are served by multiple BSs both directly and assisted by RISs.
    An HBF scheme consisting of the digital beamforming at BSs and RIS-based analog beamforming is designed.
\item  	An energy efficiency maximization problem for the HBF scheme is formulated and decomposed into two subproblems, i.e., the digital beamforming subproblem and the RIS-based analog beamforming subproblem.
    An iterative energy efficiency maximization~(EEM) algorithm is proposed to solve them.
    Moreover, the impact of the transmit power, the number of RISs, and the size of each RIS on energy efficiency is analyzed theoretically.
\item   Simulation results indicate that the RIS aided cell-free system has a better energy efficiency performance compared to traditional ones including conventional distributed antenna system~(DAS), conventional cell-free system, and the no-RIS case.
    The impact of the number of RISs, the size of each RIS, the quantization bits of discrete phase shifts on energy efficiency is also shown numerically, which verify our theoretical analysis.
\end{enumerate}

The rest of this paper is organized as follows.
In Section~\ref{sec:model}, we provide the system model of the RIS aided cell-free MIMO system.
The RIS reflection model and the channel model are presented.
In Section~\ref{HBF}, the HBF scheme is proposed for the RIS aided cell-free system.
An energy efficiency maximization problem is formulated and then decomposed into two subproblems, i.e., the digital beamforming subproblem and RIS-based analog beamforming subproblem.
An EEM algorithm is designed in Section \ref{sec:algorithm} to solve the above subproblems iteratively.
In Section~\ref{sec:analysis}, we discuss the influence of the number of RISs, the size of each RIS, and the transmit power on energy efficiency theoretically.
The complexity and convergence of the proposed EEM algorithm are also analyzed.
The simulation results are given in Section~\ref{sec:simulation} to evaluate the energy efficiency and sum rate performance to validate our analysis.
Finally, we draw our conclusions in Section~\ref{sec:conclusion}.

Scalars are denoted by italic letters, vectors and matrices are denoted by bold-face lower-case and uppercase letters, respectively.
$\overline{a}$ denotes the complex conjugate of a complex scalar $a$.
The real part and the imaginary part of $a$ are denoted by $\text{Re}\{s\}$ and $\text{Im}\{s\}$, respectively.
For a complex-valued vector $\mathbf{x}$, $\text{diag}(\mathbf{x})$ denotes a diagonal matrix whose diagonal element is the corresponding element in $\mathbf{x}$.
For a square matrix $\mathbf{S}$, $\text{Tr}(\mathbf{S})$ denotes its trace.
For any general matrix~$\mathbf{M}$, $\mathbf{M}^H$ denotes its conjugate transpose, respectively.
$\mathbf{I}$ denotes an identity matrix.
In addition, the major notations in Section~\ref{sec:model} and Section~\ref{HBF} are summarized in TABLE~\ref{tab:notation}.

\section{System Model\label{sec:model}}%

In this section, we first introduce an RIS-aided downlink cell-free MIMO communication system in which multi-antenna BSs and multiple RISs coordinate to serve various single-antenna users.
The RIS reflection model and the channel model are then given.
\begin{figure}[t]
	\centering
    \includegraphics[width=1.0\textwidth]{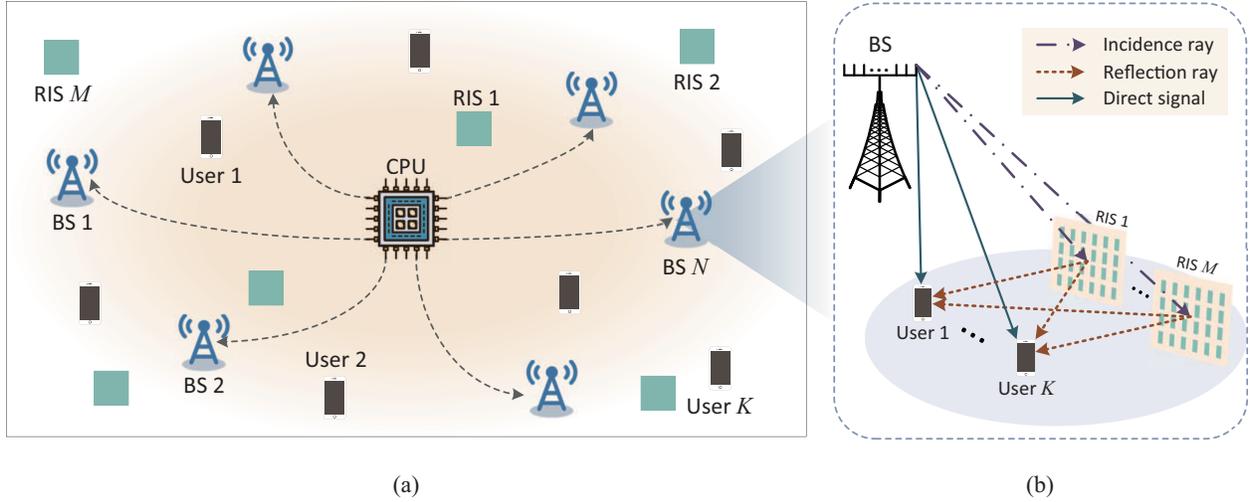}
	\caption{(a) RIS aided cell-free MIMO system; (b) Illustration of one RIS-aided cell in the MIMO system.}
	\label{Fig:model}
\end{figure}
\subsection{Scenario Description}

Consider a downlink multi-user system in Fig. \ref{Fig:model}(a) where $K$ single-antenna users are served by $N$ BSs each of which is equipped with $N_a$ antennas.
To improve the energy efficiency performance, we further consider to deploy $M$ RISs around these BSs, which passively reflect the signals from BSs and directly project to users, thereby forming an RIS aided cell-free MIMO system.
These RISs act as antenna arrays far away from the BS to improve the capacity and achieve a broader coverage in a low-cost way.
Moreover, the RIS technique enables signals to be transmitted via both the direct and reflected links, and thus, a more diverse transmission is obtained to improve the performance, and thus enhancing the cell-free communications.

In such a system, all BSs and RISs are coordinated to serve all users.
A central processing unit~(CPU) is deployed for network control and planning, and decides the transmission scheduling based on the locations of BSs, RISs, and users.
To characterize the optimal performance of the RIS aided cell-free MIMO system with discrete phase shifts, we assume that the channel state information~(CSI) of all channels, i.e., direct channel between BSs and users and reflected channel via multiple RISs, is perfectly known at the CPU based on the various channel acquisition methods discussed in~\cite{channel-estimation,channelestiD,channelestiH}.
As shown in Fig. \ref{Fig:model}(b), after receiving the signals transmitted from BSs, multiple RISs operating at the same frequency are employed to reflect these signals to the users synchronously.
Specifically, users receive both the signals transmitted from BSs directly and various reflection signals via multiple RISs.
\subsection{RIS Reflection Model}

RIS is an artificial thin film of electromagnetic and reconfigurable materials, consisting of $L$ \emph{RIS elements}, each of which is a sub-wavelength meta-material particle connected by multiple electrically controllable PIN diodes~\cite{Hu-RIS}.
Each PIN diode can be switched ON/OFF based on the applied bias voltages, thereby manipulating the electromagnetic response of RIS elements towards incident waves.
Benefit from these built-in programmable elements, RISs serve as low-cost reconfigurable phased arrays without extra active power sources which only rely on the combination of multiple programmable radiating elements to realize a desired transformation on the transmitted, received, or reflected waves.

By a $b$-bit re-programmable meta-material, each RIS element can be configured into $2^b$ possible amplitudes/phase shifts to reflect the radio wave.
Without loss of generality, we assume that each RIS element is configured to maximize the reflected signal, i.e., the amplitude of each element of each RIS is set as 1~\cite{beta-1, beta-2}.
Let $\theta_{m,l}$ denote the phase shift of the $l$-th RIS element of the $m$-th RIS, assuming that signals are reflected ideally without the hardware imperfections such as non-linearity and noise, the corresponding frequency response of each RIS element $q_{m,l}$~($1\leq m\leq M, 1\leq l\leq L$) can be expressed by
\begin{align}\label{response}
  &q_{m,l}=e^{j\theta_{m,l}},\\
  &\theta_{m,l}=\frac{i_{m,l}\pi}{2^{b-1}}, i_{m,l}\in\{0,1,...,2^b-1\}.
\end{align}
\subsection{Channel Model}

As shown in Fig. \ref{Fig:model}, each user receives signals from both $N$ BSs directly and $M$ RISs which reflect the signals sent by BSs.
The channel between each antenna of BS $n$ and each user $k$ consists of a \textbf{direct link} from BS $n$ to user $k$ and $M\times L$ \textbf{reflected links} since there are $M$ RISs each of which relies on $L$ elements to reflect the received signals.

Therefore, the equivalent channel $\mathbf{H}\in\mathbb{C}^{K\times NN_a}$ can be given by
\begin{equation}\label{channel}
    \mathbf{H} = \underbrace{\mathbf{H}_D}_{\text{Direct link}}+\ \underbrace{\mathbf{H}_{RU}^H\mathbf{Q}\mathbf{H}_{BR}}_{\text{Reflected links}},
\end{equation}
where $\mathbf{H}_D\in\mathbb{C}^{K\times NN_a}$ denotes the direct component between BSs and users, the $\mathbf{H}_{RU}^H\mathbf{Q}\mathbf{H}_{BR}$ denotes the reflected links via $M$ RISs.
Specifically, each RIS~$m$ reflects incident signal based on an equivalent $ML\times ML$ diagonal phase shift matrix~$\mathbf{Q}$ consisting of the phase shifts $\{q_{m,l}\}$.
$\mathbf{H}_{BR}\in\mathbb{C}^{ML\times NN_a}$ and $\mathbf{H}_{RU}^H\in\mathbb{C}^{K\times ML}$ represent the reflected channel between BSs and RISs and that between RISs and users, respectively.
We assume that all channels are statistically independent and follow Rician distribution.
Though we consider a downlink communication system in this paper, the results can be easily extended to an uplink case.
\section{Hybrid Beamforming and Problem Formulation}\label{HBF}

In this section, we first present an HBF scheme based on the system model introduced in Section~\ref{sec:model}.
An energy efficiency maximization problem is formulated and decomposed for the proposed HBF scheme.
\subsection{Hybrid Beamforming Scheme}

For each RIS $m$, $L$ RIS elements can be viewed as antenna elements far away from the BSs, inherently capable of realizing analog beamforming via determining the phase shifts of all RIS elements.
Moreover, BSs are responsible for the signal processing since RIS elements are passive devices without any digital processing capacity.
To create reflected waves towards preferable directions to serve multiple users, we present an HBF scheme for the proposed RIS aided cell-free MIMO system given the RIS reflection model and the channel model.
Specifically, the digital beamforming is performed at BSs while the analog beamforming is achieved by these RISs.

\subsubsection{Digital Beamforming}

For $K$ users, BSs encode $K$ different data streams via a digital beamformer~$\mathbf{V}_D\in\mathbb{C}^{NN_a\times K}$, and then up-converts the encoded signals over the carrier frequency and allocates the transmit powers.
We assume that $K\leq NN_a$.
The users' signals are sent directly through $N$ BSs each equipped with $N_a$ antennas, which can be given by
\begin{equation}\label{digital}
  \mathbf{x} = \mathbf{V}_D\mathbf{s},
\end{equation}
where $\mathbf{s}\in\mathbb{C}^{K\times1}$ denotes the incident signal vector for $K$ users.

\subsubsection{RIS-based Analog Beamforming}
In the RIS aided cell-free MIMO system, the analog beamforming is achieved by determining the phase shifts of all RIS elements, i.e., $\mathbf{Q}$.
Specifically, RIS~$m$ with $L$ elements performs a linear mapping from the incident signal vector to a reflected signal vector.
Therefore, the received signal at user $k$ can be expressed by
\begin{equation}\label{re-signal}
  z_k = \mathbf{H}_k \mathbf{V}_{D,k} s_k + \underbrace{\sum_{k'\neq k} \mathbf{H}_k \mathbf{V}_{D,k'} s_{k'}}_{\text{inter-user interference}} + \omega_k,
\end{equation}
where $\omega_k\sim\mathcal{C}\mathcal{N}(0,\sigma^2)$ is the additive white Gaussian noise.
$\mathbf{H}_k$ and $\mathbf{V}_{D,k}$ denote the $k$-th row and the $k$-th column of matrices $\mathbf{H}$ and $\mathbf{V}_D$, respectively.
\subsection{Energy Efficiency Maximization Problem Formulation}

Based on the RIS reflection model, the channel model, and the HBF scheme, the data rate of user $k$ can be given by
\begin{equation}\label{rate}
\begin{aligned}
  R_k  &= \log_2\left (1+\frac{|\mathbf{H}_k\mathbf{V}_{D,k}|^2}{\sum_{k'\neq k}|\mathbf{H}_k\mathbf{V}_{D,k'}|^2+\sigma^2}\right)\\
   &= \log_2\left (1+\frac{|\mathbf{H}_{D,k}\mathbf{V}_{D,k} + \mathbf{H}_{RU,k}^H\mathbf{Q}\mathbf{H}_{BR}\mathbf{V}_{D,k}|^2}{\sum_{k'\neq k}|\mathbf{H}_{D,k}\mathbf{V}_{D,k'} + \mathbf{H}_{RU,k}^H\mathbf{Q}\mathbf{H}_{BR}\mathbf{V}_{D,k'}|^2+\sigma^2}\right),
\end{aligned}
\end{equation}
where $\mathbf{H}_{D,k}$ and $\mathbf{H}_{RU,k}^H$ denote the $k$-th rows of matrices $\mathbf{H}_D$ and $\mathbf{H}_{RU}^H$, respectively.
The sum rate of all the users can be given by
\begin{equation}\label{rate}
  R = \sum\limits_{1\leq k\leq K} R_k.
\end{equation}

Since the signals are \emph{passively} reflected by RISs via controlling phase shifts of its elements, the power consumption at RISs only relies on the hardware cost which depends on the number of quantization bits~$b$ with no transmit power consumption~\cite{Chongwen Huang}.
Therefore, the total power consumption $\mathcal{P}$ consists of the transmit power of BSs, the hardware static power consumption of each BS $P_B^{(n)}$, each element $l$ of RIS $m$ with $b$-bit phase shifts $P_R^{(m,b)}$, and each user $k$ $P_U^{(k)}$, expressed by
\begin{equation}\label{totalpower}
  \mathcal{P} = \sum_n^N(\omega_n p_t^{(n)} + P_B^{(n)}) + \sum_{m,l}^{M,L}P_R^{(m,b)} + \sum_k^KP_U^{(k)},
\end{equation}
where $\omega_n > 0$ is a constant reflecting the power amplifier efficiency, feeder loss and other loss factors due to power supply and cooling for BS $n$.
$p_t^{(n)}=\text{Tr}(\mathbf{V}_{D,n}^H\mathbf{V}_{D,n})$ is the transmit power at BS $n$, where $\mathbf{V}_{D,n}$ denote the $n$-th block of $\mathbf{V}_{D}$ which represents digital beamformer of the $n$-th BS.

Therefore, the energy efficiency of the RIS aided cell-free system can be expressed by~\cite{bandEE}
\begin{equation}
\eta = \frac{BR}{\mathcal{P}},
\end{equation}
where $B$ denotes the transmission bandwidth.

The energy efficiency maxmization problem can be formulated as
\begin{subequations}\label{EEM}
\begin{align}
\max\limits_{\mathbf{V}_D,\mathbf{Q}}\quad &\eta, \label{FinalObj}\\
s.t. \ &p_t^{(n)}\leq P_T^{(n)}, \forall 1\leq n\leq N\\
  &q_{m,l}=e^{j\theta_{m,l}},\\
  &\theta_{m,l}=\frac{i_{m,l}\pi}{2^{b-1}}, i_{m,l}\in\{0,1,...,2^b-1\},
\end{align}
\end{subequations}
where $P_T^{(n)}$ is the transmit power budget available at each BS $n$, the digital beamformer~$\mathbf{V}_{D}$ and phase shift matrix $\mathbf{Q}$ are required to be optimized.
\subsection{Problem Decomposition}

The energy efficiency maximization problem~(\ref{EEM}) is difficult to solve due to the complicated interference items.
Moreover, it needs to jointly optimize both the digital beamformer and the RIS-based beamformer where the latter involves a large number of discrete variables, i.e., the phase shifts of each RIS element~$\theta_{m,l}$.
To solve this problem efficiently, we decouple it into two subproblems as shown below.

\subsubsection{Digital Beamforming Subproblem}

Given the fixed RIS-based beamformer~$\mathbf{Q}$, the digital beamforming subproblem can be formulated by
\begin{subequations}
\begin{align}
\max\limits_{\mathbf{V}_D}\quad &\eta, \label{subpro1}\\
s.t. \ &p_t^{(n)}\leq P_T^{(n)}, \forall 1\leq n\leq N.
\end{align}
\end{subequations}

\subsubsection{RIS-based Analog Beamforming Subproblem}

Similarly, the RIS-based analog beamforming subproblem with fixed digital beamformer~$\mathbf{V}_D$ can be written by
\begin{subequations}
\begin{align}
\max\limits_{\mathbf{Q}}\quad &\eta, \label{subpro2}\\
s.t. \  &q_{m,l}=e^{j\theta_{m,l}},\\
  &\theta_{m,l}=\frac{i_{m,l}\pi}{2^{b-1}}, i_{m,l}\in\{0,1,...,2^b-1\}.
\end{align}
\end{subequations}

\section{Energy Efficiency Maximization Algorithm Design\label{sec:algorithm}}%

In this section, we develop an EEM algorithm to obtain a suboptimal solution of the energy efficiency maximization problem~(\ref{EEM}) via solving the digital beamforming subproblem and RIS-based analog beamforming subproblem in an iterative manner.
Finally, we summarize the overall algorithm.
\subsection{Digital Beamforming Design}\label{sec:digital}
Cell-free MIMO can be implemented with simple linear processing such as conjugate beamforming~(CB)~\cite{CB} and zero-forcing(ZF)~\cite{ZF-result}.
Since the CB technique suffers from high inter-user interference~\cite{dbeam-ZF}, without loss of generality~\cite{Chongwen Huang,ZF-EE}, we consider ZF beamforming with power allocation as the digital beamformer to manage the inter-user interference to achieve a near-optimal solution in MIMO systems.
According to the results in \cite{ZF-result}, the digital beamformer with a fixed RIS based analog beamformer can be given as
\begin{equation}\label{ZF beamforming}
  \mathbf{V}_D = \mathbf{H}^H(\mathbf{H}\mathbf{H}^H)^{-1}\mathbf{P}^{\frac{1}{2}} =\widetilde{\mathbf{V}}_D\mathbf{P}^{\frac{1}{2}},
\end{equation}
where $\widetilde{\mathbf{V}}_D = \mathbf{H}^H(\mathbf{H}\mathbf{H}^H)^{-1}$ and $\mathbf{P}\triangleq\text{diag}(p_1,...,p_K)$ is the power matrix in which the $k$-th diagonal element denotes the transmit power allocated to the signal intended for the $k$-th user from all the BSs.

Note that the ZF beamformer has the following properties:
\begin{equation}\label{ZF-proper}
\begin{aligned}
  &|\mathbf{H}_k\mathbf{V}_{D,k}| = \sqrt{p_k},\\
  &|\mathbf{H}_k\mathbf{V}_{D,k'}| = 0,\forall k'\neq k.
\end{aligned}
\end{equation}

For convenience, we assume that each BS has the same transmit power amplifier efficiency, i.e., $\omega_1 = ... = \omega_n = \omega$.
Based on these properties of ZF beamforming, subproblem~(\ref{subpro1}) can be reduced to a power allocation problem which can be given by
\begin{subequations}\label{powerpro}
\begin{align}
  \max\limits_{\{p_k\geqslant 0\}}\quad &\frac{\sum_{k=1}^{K}\log_2(1+p_k\sigma^{-2})}{\omega\sum_{k=1}^{K}p_k + \mathcal{P}_s},\label{dpro}\\
s.t.  \ &\text{Tr}(\mathbf{P}^{\frac{1}{2}}\widetilde{\mathbf{V}}_{D,n}^H\widetilde{\mathbf{V}}_{D,n}\mathbf{P}^{\frac{1}{2}})\leq P_T^{(n)}, \forall 1\leq n\leq N, \label{powercon}.
\end{align}
\end{subequations}
where $\mathcal{P}_s = \sum_n^NP_B^{(n)} + \sum_{m,l}^{M,L}P_R^{(m,b)} + \sum_k^KP_U^{(k)}$ is the hardware static power consumption which is independent of the optimization variables in problem~(\ref{EEM}), i.e., $\mathbf{V}_{D}$ and $\mathbf{Q}$.
$\widetilde{\mathbf{V}}_{D,n}$ denote the $n$-th block of $\widetilde{\mathbf{V}}_{D}$.

\begin{lemma}
By adopting Benson's transform method~\cite{benson}, the energy efficiency maximization problem in~(\ref{EEM}) is equivalent to
\begin{subequations}\label{digitalfinal}
\begin{align}
  \max\limits_{\{p_k\geqslant 0\}}\quad &\eta(\mathbf{P},y)=2y[\sum_{k=1}^{K}\log_2(1+p_k\sigma^{-2})]^{\frac{1}{2}}- y^2(\omega\sum_{k=1}^{K}p_k + \mathcal{P}_s),\\
s.t.  \
&\text{Tr}(\mathbf{P}^{\frac{1}{2}}\widetilde{\mathbf{V}}_{D,n}^H\widetilde{\mathbf{V}}_{D,n}\mathbf{P}^{\frac{1}{2}})\leq P_T^{(n)}, \forall 1\leq n\leq N,
\end{align}
\end{subequations}
where the optimal $y^*$ can be obtained based on the results in~\cite{weiyu-digital}, expressed by
\begin{equation}\label{y}
  y^* = \frac{\sum_{k=1}^{k}\log_2(1+p_k\sigma^{-2})}{\omega\sum_{k=1}^{K}p_k + \mathcal{P}_s}
\end{equation}
\end{lemma}

As such, problem~(\ref{powerpro}) is transformed into a convex problem of $\mathbf{P}$, and thus, the digital beamforming subproblem can be solved by existing convex optimization techniques.
The algorithm can be summarized in Algorithm~\ref{alg:digital}.
\begin{algorithm}[t]
\label{alg:digital}
\caption{Digital Beamforming Design Algorithm}
\LinesNumbered
\KwIn{RIS-based beamformer matrix~$\mathbf{Q}$}
\KwOut{Digital beamformer matrix~$\mathbf{V}_D$}
{
Initialize $\mathbf{P}$ to a feasible value\;
\For{each iteration}{
Update $y$ by~(\ref{y})\;
Update $\mathbf{P}$ by solving the convex optimization problem~(\ref{digitalfinal}) with fixed $y$.
}
Derive the digital beamforming matrix from the optimal power allocation solution based on (\ref{ZF beamforming}).
}
\end{algorithm}
\subsection{RIS-based Analog Beamforming Design}

Note that we optimize the energy efficiency via performing the digital beamforming and RIS-based analog beamforming iteratively.
Based on the designed digital beamforming as shown in~(\ref{ZF beamforming}), the energy efficiency in~(\ref{powerpro}) depends on the RIS-based analog beamformer~$\mathbf{Q}$ only through the power constraint~(\ref{powercon}).
Therefore, the RIS-based analog beamforming subproblem can be reformulated as a power minimization problem, expressed by
\begin{subequations}
\begin{align}
  \min\limits_{\theta_{m,l}}\quad &f(\mathbf{Q}), \label{powermin}\\
s.t. \quad&q_{m,l}=e^{j\theta_{m,l}},\\
  &\theta_{m,l}=\frac{i_{m,l}\pi}{2^{b-1}}, i_{m,l}\in\{0,1,...,2^b-1\}.
\end{align}
\end{subequations}
where
\begin{equation}\label{f_Q}
  f(\mathbf{Q}) = \text{Tr}(\mathbf{V}_D^H\mathbf{V}_D).
\end{equation}
However, problem~(\ref{powermin}) is still difficult to solve due to the large-scale of~$\mathbf{Q}$.
For simplicity, we ignore the items about direct link~$\mathbf{H}_D$ in~$f(\mathbf{Q})$ which is unrelated to the optimization objective~$\mathbf{Q}$.
For $K=NN_a$, $f(\mathbf{Q})$ can be rewritten by
\begin{equation}\label{f_Q_simple}
\begin{aligned}
  f(\mathbf{Q})&=\text{Tr}(\mathbf{V}_D^H\mathbf{V}_D)\\
  &=\text{Tr}[(\mathbf{P}^{-\frac{1}{2}}\mathbf{H}_{RU}^H\mathbf{Q}\mathbf{H}_{BR}\mathbf{H}_{BR}^H\mathbf{Q}^H\mathbf{H}_{RU}\mathbf{P}^{-\frac{1}{2}})^{-1}].
\end{aligned}
\end{equation}

For brevity, we define the index $j = (m-1)L+l$, and thus~$q_{m,l}$ is abbreviated to $q_j$.
To separate the optimization variable~$q_j$ from other fixed elements, we set the $j$-th diagonal element of $\mathbf{Q}$ as $0$ and use $\mathbf{Q}^{(-j)}$ to denote the new matrix.
Hence, $f(\mathbf{Q})$ can be rewritten as
\begin{equation}
\begin{aligned}
  f(\mathbf{Q})&=\text{Tr}[((\mathbf{A}_j+\mathbf{B}_jq_j)(\mathbf{A}_j^H+\mathbf{B}_j^H\overline{q}_j))^{-1}]\\
  &=\text{Tr}[(\mathbf{D}_j+\mathbf{C}_j\overline{q}_j+\mathbf{C}_j^Hq_j)^{-1}],
\end{aligned}
\end{equation}
where $\mathbf{A}_j=\widetilde{\mathbf{P}}\mathbf{Q}^{(-j)}\mathbf{H}_{BR}$ and $\mathbf{B}_j=\widetilde{\mathbf{P}}^{(j)}\mathbf{H}_{BR}^{(j)}$.
We denote that $\widetilde{\mathbf{P}}=\mathbf{P}^{-\frac{1}{2}}\mathbf{H}_{RU}^H$.
The $j$-th column and $j$-th row of $\widetilde{\mathbf{P}}$ and $\mathbf{H}_{BR}$ are represented by $\widetilde{\mathbf{P}}^{(j)}$ and $\mathbf{H}_{BR}^{(j)}$, respectively.
Moreover, $\mathbf{C}_j=\mathbf{A}_j\mathbf{A}_j^H+\mathbf{B}_j\mathbf{B}_j^H$ and $\mathbf{D}_j=\mathbf{A}_j+\mathbf{A}_j^H$.

Note that $\mathbf{C}_j\overline{q}_j$ and $\mathbf{C}_j^Hq_j$ are one-rank matrices and $\mathbf{D}_j$ is a full-rank matrix.
Therefore, $f(\mathbf{Q})$ can be rewritten according the Sherman Morrison formula~\cite{Sherman}, i.e., $(\mathbf{A}+\mathbf{B})^{-1}=\mathbf{A}^{-1}-\frac{\mathbf{A}^{-1}\mathbf{B}\mathbf{A}^{-1}}{1+\text{Tr}(\mathbf{A}^{-1}\mathbf{B})}$.
Through simplification, $f(\mathbf{Q})$ can be expressed by
\begin{equation}\label{f_Q_de}
  f(\mathbf{Q})=\left(\frac{a_1e^{3j\theta_{m,l}}+a_2e^{2j\theta_{m,l}}+a_3e^{j\theta_{m,l}}+a_4}{a_5e^{3j\theta_{m,l}}+a_6e^{2j\theta_{m,l}}+a_7e^{j\theta_{m,l}}+a_8}\right),
\end{equation}
where $a_1\sim a_8$ are defined as in Appendix A and independent of $\theta_{m,l}$.

The minimum value of $f(\mathbf{Q})$ can be obtained with respect to $\theta_{m,l}$ which satisfies
\begin{equation}
  \frac{\partial f(\mathbf{Q})}{\partial \theta_{m,l}}=0.
\end{equation}
Based on the results in Appendix B, the optimal value of $\theta_{m,l}$ can be obtained by
\begin{equation}\label{chi}
  \theta_{m,l}^*=2\arctan\chi,
\end{equation}
where $\chi$ is defined as in Appendix B.
Since $0\leq\theta_{m,l}\leq2\pi$, only the smaller of the two solutions of $\arctan\chi$, i.e., the solution which is less than~$\pi$, is taken into account.
Moreover, since only a limited number of discrete phase shifts are available at the RIS-based analog beamforming, the optimal value $\theta_{m,l}^*$ should be quantized to the nearest points in the set $\mathcal{F}=\{\frac{i_{m,l}\pi}{2^{b-1}}\}, i_{m,l}\in\{0,1,...,2^b-1\}$.

Therefore, starting from a randomly initiated RIS-based analog beamformer, the optimal~$\mathbf{Q}$ can be obtained by sequentially updating the phase shift of each RIS element~$(m,l)$ in an iterative manner, until the algorithm converges to a local minimum of~$f(\mathbf{Q})$.
The algorithm is summarized in Algorithm~\ref{alg:RIS}.

\begin{algorithm}[t]
\label{alg:RIS}
\caption{RIS-based Analog Beamforming Algorithm}
\LinesNumbered
\For{each iteration}
{
\For{$\mathcal{J}=1\rightarrow M\times L$}
{
Calculate $\chi$ according to Appendix B\;
Obtain $\theta_{m,l}^*=2\arctan\chi$\;
Quantize $\theta_{m,l}^*$ to the nearest points in the feasible set $\mathcal{F}$\;
Set $q_{m,l}=e^{j\theta_{m,l}^*}$.
}
}
\end{algorithm}
\subsection{Overall Algorithm Description}
In this section, we propose the EEM algorithm to solve the energy efficiency maximization problem~(\ref{EEM}) in an iterative manner.
We first design the digital beamforming $\mathbf{V}_D$ by Algorithm~\ref{alg:digital} with fixed RIS-based analog beamforming.
Based on the achieved digital beamforming, the RIS-based analog beamforming is optimized by Algorithm~\ref{alg:RIS}.
Let $\mathbf{\eta}$ denote the objective function of the energy efficiency maximization problem~(\ref{EEM}).
The EEM algorithm converges if, in the $\zeta$-th iteration, the value difference of the objective functions between two adjacent iterations is less than a predefined threshold $\varepsilon$, i.e., $\mathbf{\eta}^{(\zeta)}-\mathbf{\eta}^{(\zeta-1)}\leq\varepsilon$.

\section{Theoretical Analysis of RIS Aided Cell-free System}\label{sec:analysis}

In this section, we first analyze the convergence and computational complexity of the proposed EEM algorithm, then the impact of the transmit power, number of RISs, and size of each RIS on the energy efficiency of the RIS aided cell-free system are explore theoretically.
\subsection{Properties of the Energy Efficiency Maximization Algorithm}
We now analyze the convergence and computational complexity of the EEM algorithm.
\subsubsection{Convergence}

In the $\zeta$-th iteration, by performing Algorithm~\ref{alg:digital}, a better digital beamforming~$\mathbf{V}_D^{(\zeta)}$ is achieved with a fixed RIS-based analog beamforming~$\mathbf{Q}^{(\zeta-1)}$.
Hence, we have
\begin{equation}
  \eta(\mathbf{V}_D^{(\zeta)}, \mathbf{Q}^{(\zeta-1)})\geqslant\eta(\mathbf{V}_D^{(\zeta-1)}, \mathbf{Q}^{(\zeta-1)}).
\end{equation}
Similarly, given digital beamforming~$\mathbf{V}_D^{(\zeta)}$, the RIS-based analog beamforming is optimized to maximize the energy efficiency by Algorithm~\ref{alg:RIS}, expressed by
\begin{equation}
\eta(\mathbf{V}_D^{(\zeta)}, \mathbf{Q}^{(\zeta)})\geqslant\eta(\mathbf{V}_D^{(\zeta)}, \mathbf{Q}^{(\zeta-1)}).
\end{equation}
Based on the above inequalities, we can obtain that
\begin{equation}
\eta(\mathbf{V}_D^{(\zeta)}, \mathbf{Q}^{(\zeta)})\geqslant\eta(\mathbf{V}_D^{(\zeta-1)}, \mathbf{Q}^{(\zeta-1)}).
\end{equation}
This indicates that the objective value of problem~(\ref{EEM}) is non-decreasing after each iteration, and thus, the proposed EEM algorithm is guaranteed to converge.

\subsubsection{Computational Complexity}
We now analyze the computational complexity of the proposed EEM algorithm for two subproblems separately.
\begin{itemize}
  \item \emph{Digital beamforming:} According to Algorithm~\ref{alg:digital}, the received power for each user should be optimized by solving a convex problem which has a polynomial complexity in the number of optimization variables, i.e., the number of users $K$.
      Therefore, its computational complexity is $\mathcal{O}(K^t)$, where $1\leq t \leq 4$~\cite{complexity-convex}.
  \item \emph{RIS-based analog beamforming:} In each iteration, the phase shift of each RIS element~$(m,l)$ is updated sequentially to obtain the optimal~$\mathbf{Q}$ according to the closed-form expression in Algorithm~\ref{alg:RIS} which is \emph{unrelated} to the quantization bits for discrete phase shifts~$b$.
      Therefore, the complexity of Algorithm~\ref{alg:RIS} is $\mathcal{O}(ML)$ in each iteration while that of exhaustive search method is $\mathcal{O}(2^{bML})$.
\end{itemize}
\subsection{Performance Analysis of RIS Aided Cell-free System}
In this part, the impact of the transmit power, number of RISs, and size of each RIS on the energy efficiency of the RIS aided cell-free system is analyzed, respectively.
\subsubsection{Impact of the transmit power}
As given in the energy efficiency maximization problem~(\ref{EEM}), the transmit power is one of the key optimization variables.
In this part, we analyze the impact of transmit power on energy efficiency in the RIS aided cell-free system.

\begin{prop}\label{prop:power}
  The energy efficiency grows rapidly with a low transmit power budget available at each BS $n$, i.e., $P_T^{(n)}$, and gradually flattens as $P_T^{(n)}$ is large enough.
\end{prop}
\begin{proof}
Based on problem (\ref{powerpro}), the energy efficiency $\eta$ at \emph{high-SNR} can be rewritten as
\begin{equation}
  \eta \approx \frac{\sum_{k=1}^{K}\log_2(p_k\sigma^{-2})}{\omega\sum_{k=1}^{K}p_k + \mathcal{P}_s}
  \overset{(a)}\approx  \frac{\log_2(\sum_{k=1}^{K}p_k\sigma^{-2})}{\omega\sum_{k=1}^{K}p_k + \mathcal{P}_s}
  =\frac{\log_2(\sum_{k=1}^{K}p_k)-\log_2(\sigma^{2})}{\omega\sum_{k=1}^{K}p_k + \mathcal{P}_s},
\end{equation}
where (a) is obtained by the property of the logarithmic function~\cite{liangletter}.
Hence, the derivative of energy efficiency $\eta$ with respect to $\sum_{k=1}^{K}p_k$ can be given by
\begin{equation}\label{dpower}
  \frac{\partial\eta}{\partial(\sum_{k=1}^{K}p_k)}=
  \frac{\frac{\omega B}{\ln2}\mathcal{P}_s+\omega B\log_2\sigma^2\sum_{k=1}^{K}p_k+[\sum_{k=1}^{K}p_k-\sum_{k=1}^{K}p_k\ln(\sum_{k=1}^{K}p_k)]\frac{\omega B}{\ln2}}
  {\mathcal{P}_s^2\sum_{k=1}^{K}p_k+2\omega\mathcal{P}_s(\sum_{k=1}^{K}p_k)^2+\omega^3(\sum_{k=1}^{K}p_k)^3}.
\end{equation}
It is easy to find that the denominator is positive.
We define the molecule as $h(p_k)$, when the transmit power tends to zero, $h(p_k)$ is also positive, expressed by
\begin{equation}
  \lim\limits_{\sum_{k=1}^{K}p_k\rightarrow 0}h(p_k)=\frac{\omega B}{\ln2}\mathcal{P}_s>0.
\end{equation}

The derivative of $h(p_k)$ with respect to $\sum_{k=1}^{K}p_k$ can be given by
\begin{equation}
  h'(p_k)=-\omega B\log_2\frac{\sum_{k=1}^{K}p_k}{\sigma^2},
\end{equation}
where $\log_2\frac{\sum_{k=1}^{K}p_k}{\sigma^2}$ is approximately equal to the sum rate which is positive, thus, $h'(p_k)<0$.

Therefore, the equation $\frac{\partial\eta}{\partial(\sum_{k=1}^{K}p_k)}=0$ has an unique solution $\sum_{k=1}^{K}p_k^*$, and satisfies that
\begin{equation}
\frac{\partial\eta}{\partial(\sum_{k=1}^{K}p_k)}
\begin{cases}
>0, & \text{if }\sum_{k=1}^{K}p_k<\sum_{k=1}^{K}p_k^*\\
<0, & \text{if }\sum_{k=1}^{K}p_k>\sum_{k=1}^{K}p_k^*.
\end{cases}
\end{equation}

Hence, with a low transmit power budget available at each BS $n$, i.e., $P_T^{(n)}$, the energy efficiency grows with $P_T^{(n)}$.
However, when $P_T^{(n)}$ is large enough to achieve $\sum_{k=1}^{K}p_k=\sum_{k=1}^{K}p_k^*$, each $p_k$ is no longer growing as $P_T^{(n)}$ continues to increase, and thus, the energy efficiency gradually flattens.
This completes the proof.
\end{proof}

\subsubsection{Impact of the number of RISs}
To improve the energy efficiency performance, \emph{multiple} RISs are deployed to create favorable propagation conditions via configurable reflection from BSs to users to enhance the cell-free communication.
Therefore, it is important to analyze the impact of the number of RISs on energy efficiency given a fixed total number of BSs and RISs.

By ignoring the items about direct link~$\mathbf{H}_D$ which is unrelated to the number of RISs~$M$, the transmit power allocated to the signals intended for the user $k$ can be expressed by
\begin{equation}\label{powerk}
\begin{aligned}
  p_k &= |\mathbf{H}_k\mathbf{V}_{D,k}|^2\\
  &=\mathbf{V}_{D,k}^H\mathbf{H}_{BR}^H\mathbf{Q}^H\mathbf{H}_{RU,k}\mathbf{H}_{RU,k}^H\mathbf{Q}\mathbf{H}_{BR}\mathbf{V}_{D,k}\\
  &=\mathbf{V}_{D,k}^H\mathbf{H}_{BR}^H\mathbf{Q}^H\text{Tr}(\mathbf{H}_{RU,k}^H\mathbf{H}_{RU,k})\mathbf{Q}\mathbf{H}_{BR}\mathbf{V}_{D,k}\\
  &\overset{(a)}\approx M^2L^2\mathbf{V}_{D,k}^H\mathbf{V}_{D,k},
\end{aligned}
\end{equation}
where (a) is obtained by the well-known \emph{channel hardening} effect in MIMO communication systems~\cite{water-filling}.
Specifically, as the number of receive~(or transmit) antennas grow large while keeping the number of transmit~(or receive) antennas constant, the column-vectors~(or row-vectors) of the propagation matrix are asymptotically orthogonal.
Hence, we have
\begin{equation}\label{hardening}
\begin{aligned}
  &\lim_{M\times L\rightarrow \infty}\mathbf{H}_{BR}^H\mathbf{H}_{BR}\approx (M\times L)\mathbf{I},\\
  &\lim_{M\times L\rightarrow \infty}\mathbf{H}_{RU}^H\mathbf{H}_{RU}\approx (M\times L)\mathbf{I},
\end{aligned}
\end{equation}
where $\mathbf{I}$ denotes the identity matrix.
Moreover, the phase shift matrix~$\mathbf{Q}$ satisfy $\mathbf{Q}\mathbf{Q}^H=\mathbf{I}$ since $\mathbf{Q}$ is a diagonal matrix where the module of each element is $1$, i.e., $|q_{m,l}|=1,\forall m,l$.
Therefore, approximate formula $(a)$ in (\ref{powerk}) can be obtained.
Based on (\ref{powerk}), we present the following proposition, which will be proved in Appendix~\ref{appC}.
\begin{prop}\label{prop:M}
  When a large amount of RIS elements are deployed to assist the cell-free system, i.e., $M\times L \rightarrow\infty$, the energy efficiency decreases as the number of RISs grows.
\end{prop}

\subsubsection{Impact of the size of each RIS}
Each RIS relies on the combination of multiple programmable radiating elements to realize a desired transformation on the transmitted, received, or reflected waves~\cite{size}.
We now reveal the impact of the size of each RIS, i.e., the number of elements of each RIS $L$, on the energy efficiency performance.

\begin{lemma}\label{lemmaL}
  The energy efficiency of the RIS aided cell-free system tends to zero when size of each RIS tends to infinity, i.e., $\lim\limits_{L\rightarrow \infty}\eta=0$.
\end{lemma}
\begin{proof}
  When the size of each RIS tends to infinity, i.e., $L\rightarrow \infty$, based on~(\ref{powerpro}) and (\ref{powerk}), the energy efficiency $\eta$ at \textbf{high SNR} can be rewritten as
\begin{equation}\label{EE-size}
  \eta\approx\frac{\sum_{k=1}^{K}\log_2(p_k\sigma^{-2})}{\omega\sum_{k=1}^{K}p_k + \mathcal{P}_s}
  \approx\frac{2BK\log_2L+B\sum_{k=1}^{K}log_2(M^2\mathbf{V}_{D,k}^H\mathbf{V}_{D,k}\sigma^{-2})}{\omega L^2\sum_{k=1}^{K} (M^2\mathbf{V}_{D,k}^H\mathbf{V}_{D,k})+L\cdot MP_R+NP_B+KP_U}.
\end{equation}
It is obvious that the energy efficiency tends to zero when size of each RIS tends to infinity, i.e., $\lim\limits_{L\rightarrow \infty}\eta=0$.
\end{proof}
Based on Lemma~\ref{lemmaL}, we present the following proposition, which will be proved in Appendix~\ref{app:size}.
\begin{prop}\label{prop:size}
  As the size of each RIS $L$ increases, the energy efficiency of RIS aided cell-free system increases first and then gradually drops to zero if $\frac{\mathcal{P}_s}{M\cdot P_R}>\ln{\frac{NP_T}{k\sigma^2}}$.
\end{prop}
\renewcommand\arraystretch{1.4}
\begin{table}[t]
\small
    \centering
    \caption{SIMULATION PARAMETERS}
    \begin{tabular}{c|c}
        \hline
        \hline
        \textbf{Parameter} & \textbf{Value}\\
        \hline
        \hline
        Rician factor~$\kappa$ & 4\\
        \hline
        Number of BSs $N$   &   4\\
        \hline
        Number of antennas of each BS $N_a$   & 8\\
        \hline
        Number of users $K$ &   8\\
        \hline
        Noise power $\sigma^2$  & $-90$~dBm\\
        \hline
        Carrier frequency &   5.9GHz\\
        \hline
        Size of each RIS element    &   0.02m\\
        \hline
        Hardware static power consumption of each BS    &   10dBw\\
        \hline
        Hardware static power consumption of each RIS element with 1-bit phase shifts    &   5dBm\\
        \hline
        Hardware static power consumption of each RIS element with 2-bit phase shifts    &   10dBm\\
        \hline
        Hardware static power consumption of each RIS element with 3-bit phase shifts    &   15dBm\\
        \hline
        Hardware static power consumption of each RIS element with continuous phase shifts    &   25dBm\\
        \hline
        Hardware static power consumption of each distributed antenna in DAS    &   20dBm\\
        \hline
        Hardware static power consumption of each user    &   10dBm\\
        \hline
        Thresholds $\varrho$ and $\varepsilon$ & 0.001\\
        \hline
        \hline
    \end{tabular}   \label{simulation}
\end{table}
\section{Simulation Results\label{sec:simulation}}%
In this section, we evaluate our proposed EEM algorithm for the RIS aided cell-free system in terms of energy efficiency.
We show how the energy efficiency is influenced by the transmit power per BS, the number of RISs, the size of each RISs, and the number of quantization bits for discrete phase shifts.
For comparison, the following schemes are considered as benchmarks.
\begin{enumerate}
  \item \emph{Conventional distributed antenna system~\cite{DAS}}: Multiple distributed antennas are deployed remotely rather than centrally at the BSs.
      For a fair comparison, we assume that the distributed antennas in the DAS and the RIS elements in the proposed scheme have the same number and locations since each RIS element can be viewed as an antenna far away from BSs.
  \item \emph{No-RIS case}: Conventional user-centric wireless communication networks without RISs. $K$ single-antenna users are served by $N$ BSs simultaneously.
  \item \emph{Conventional cell-free system~\cite{cell-free,cell-free1}:} Conventional cell-free communication networks without RISs. Different from the no-RIS case, we replace all RISs by BSs, i.e., the number of BSs is set as $N+M$.
\end{enumerate}
\begin{figure}[t]
\centering
\includegraphics[width=0.65\textwidth]{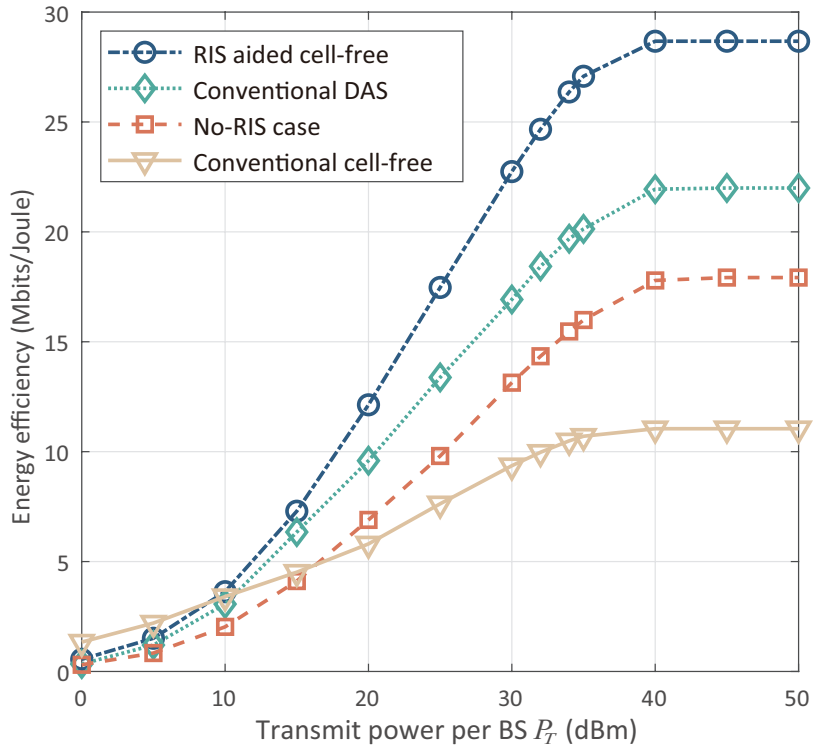}
\caption{Energy efficiency vs. transmit power per BS of different schemes with $M = 3$, $L = 64$, and $b = 3$.}
\label{Fig:SNR}
\end{figure}

For the propagations, we use the UMa path loss model in \cite{Uma} as the distance-dependent channel path loss model.
BSs and RISs are uniformly distributed in circles around the users within a radius of $150$m.
Users are uniformly distributed in the circle of radius $20$m.
The distance between two adjacent antennas at the BS is $1$m.
For small-scale fading, we assume the Rician fading channel model for all channels involved.
Simulation parameters are set up based on the existing works~\cite{Di-RIS, liangletter, Chongwen Huang}, as given in TABLE~\ref{simulation}.
\begin{figure}[t]
\centering
\includegraphics[width=0.65\textwidth]{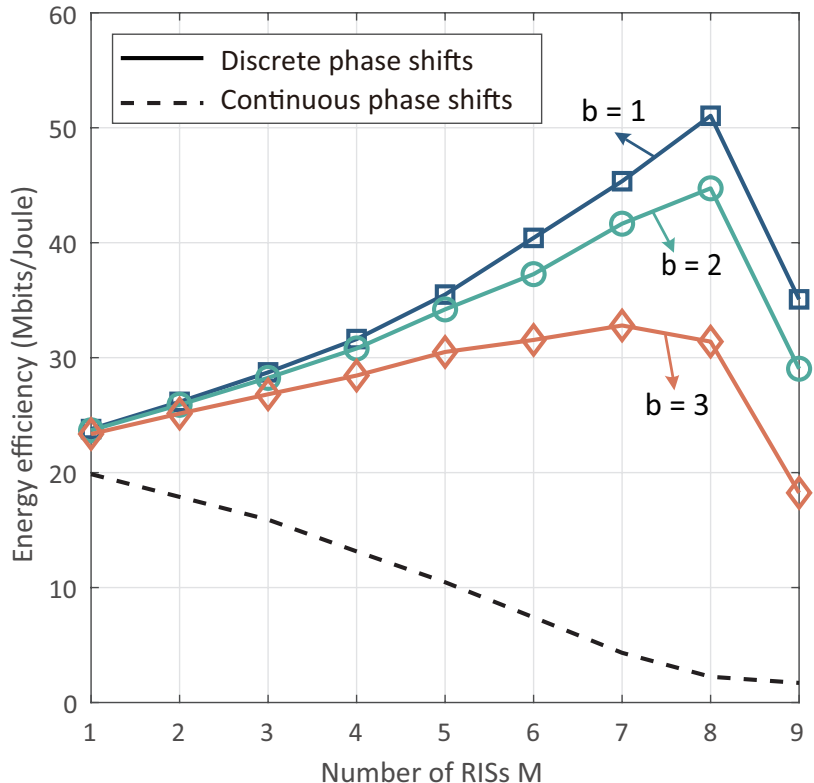}
\caption{Energy efficiency vs. number of RISs~$M$ for different number of quantization bits~$b$~($P_T = 30$~dBm).}
\label{Fig:M}
\end{figure}
\subsection{Comparison with Benchmarks}
Fig. \ref{Fig:SNR} shows the energy efficiency of the RIS aided cell-free MIMO system versus the transmit power per BS~$P_T$, obtained by different algorithms with $3$ RISs~(i.e., $M = 3$).
Each RIS consists of 64 elements~(i.e., $L = 64$), the number of quantization bits for discrete phase shifts~$b$ is set as $3$.
We observe that the energy efficiency of different systems grows rapidly with a low transmit power per BS $P_T$ and gradually flattens as $P_T$ continues to increase, as proved in Proposition~\ref{prop:power}.

Compared with the no-RIS case, we observe that these RISs deployed in the cell-free system can effectively improve the energy efficiency of the multi-user communication system.
Compared with the conventional DAS, it can be observed that the energy efficiency performance is improved in the proposed system since signals are transmitted via a larger number of independent paths~(i.e., the direct and reflected links), implying that the spatial resources are better utilized.
However, with a low transmit power per BS~$P_T$, the energy efficiency of conventional cell-free systems is higher than that of the proposed scheme.
This is because with a lower total transmit power budget available at all BSs, the energy efficiency mainly depends on the sum rate.
Conventional cell-free systems provide more transmit power by more BSs, thereby achieving a higher sum rate.
When the transmit power per BS is larger than $10$dBm, the hardware static power consumption becomes one of the dominant items of energy efficiency, and thus, our proposed system can achieve a better energy efficiency performance than the conventional cell-free systems.
\subsection{Impact of Number of RISs}
\begin{figure}[t]
\centering
\includegraphics[width=0.65\textwidth]{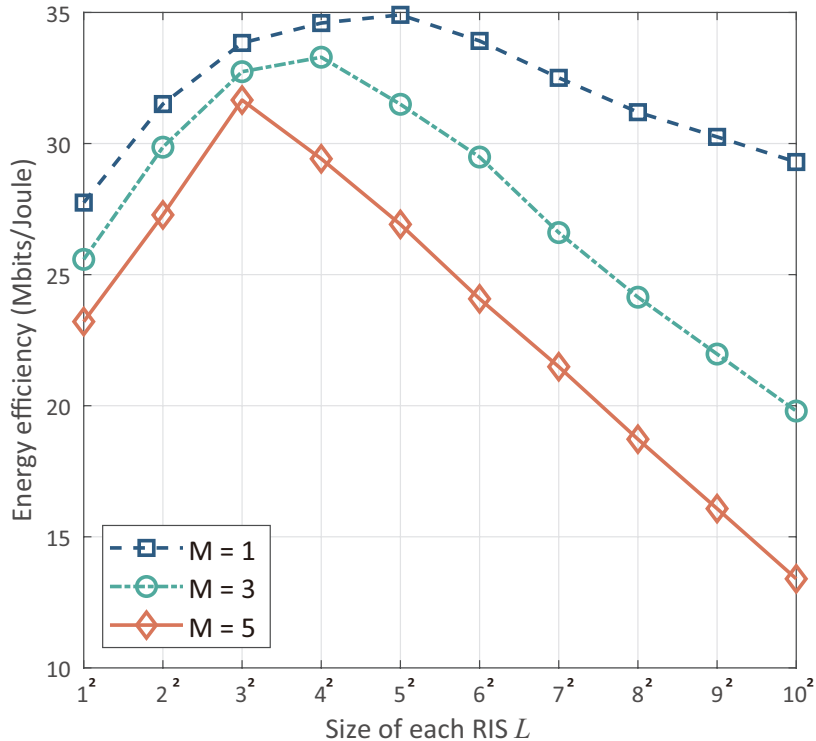}
\caption{Energy efficiency vs. size of RISs~$L$ for different number of RISs~$M$~($P_T = 30$~dBm).}
\label{Fig:L}
\end{figure}
\begin{figure}[t]
\centering
\includegraphics[width=0.65\textwidth]{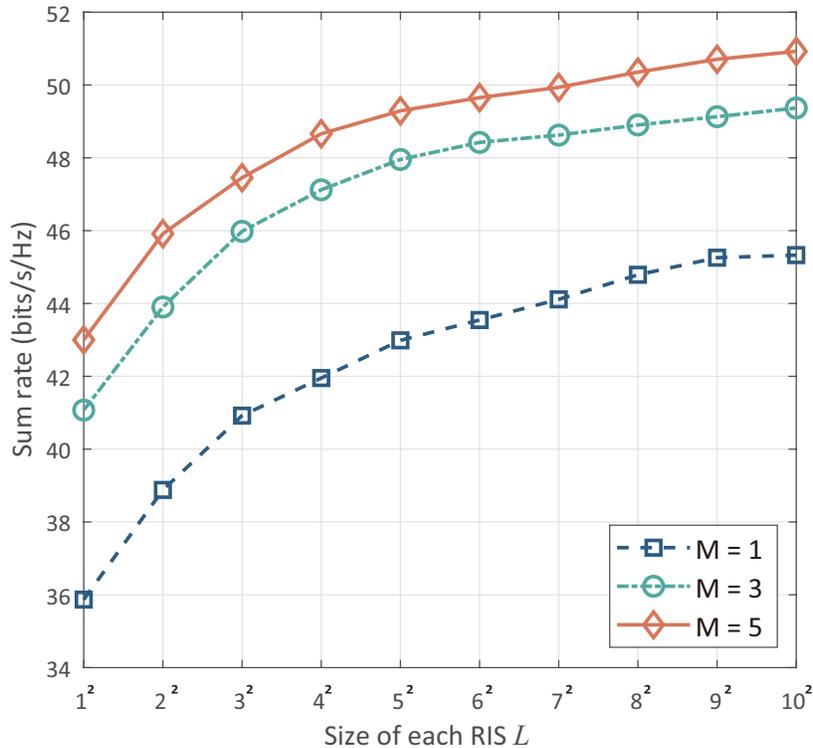}
\caption{Energy efficiency vs. size of RISs~$L$ for different number of RISs~$M$~($P_T = 30$~dBm).}
\label{Fig:Lrate}
\end{figure}

To fully explore how RISs as cost-efficient devices influence the energy efficiency of the cell-free system, we evaluate the performance with different number of RISs given a fixed total number of BSs and RISs, i.e., $N_0 = N + M$ is a constant.
Specifically, Fig. \ref{Fig:M} shows the energy efficiency vs. the number of RISs~$M$ with $N_0 = 10$, $P_T = 30$~dBm, $L = 64$, and different number of quantization bits~$b$.
As the quantization bits of RISs $b$ increases, the hardware static power consumption of each RIS increases, as given in TABLE~\ref{simulation}.

It can be seen that the energy efficiency increases with the number RISs $M$ for a small-scale deployment of RISs.
This figure also implies that the slope of energy efficiency vs. the number of deployed RISs decreases with the quantization bits of RISs $b$ since the higher resolution of RISs brings a significantly higher hardware static power consumption.
As the number of RISs $M$ continues to increase, i.e., $M\times L\rightarrow\infty$, the energy efficiency decreases as proved in Proposition~\ref{prop:M}.
The optimal number of deployed RISs $M$ shows up at a smaller value with a higher level quantization bits of RISs $b$, especially with the continuous phase shifts, due to the higher energy consumption.
\subsection{Impact of Size of each RIS}

Fig. \ref{Fig:L} shows the impact of the RIS size $L$ on the energy efficiency performance with $P_T = 30$~dBm and $b = 3$.
We can observe that as the number of elements of each RIS increases, the energy efficiency first increases and then decreases since the hardware static power consumption of each RIS increases with the number of elements in each RIS $L$, which verifies our theoretical analysis in the Proposition~\ref{prop:size}.
The optimal RIS size~$L$ shows up at a smaller value when more RISs are deployed in the cell-free system.

Fig. \ref{Fig:Lrate} depicts the sum rate versus the size of each RIS $L$ with different number of RISs in the cell-free system.
The transmit power per BS $P_T$ is set as $30$~dBm and the number of quantization bits for discrete phase shift $b$ is set as $3$.
As the number of elements of each RIS increases, the sum rate grows and gradually converges to a stable value.
Moreover, when the number of RISs grows, the gap between the curves obtained with different size of each RIS shrinks since a larger scale of RIS deployment usually provides more freedom of generating directional beams.
\begin{figure}[t]
\centering
\includegraphics[width=0.65\textwidth]{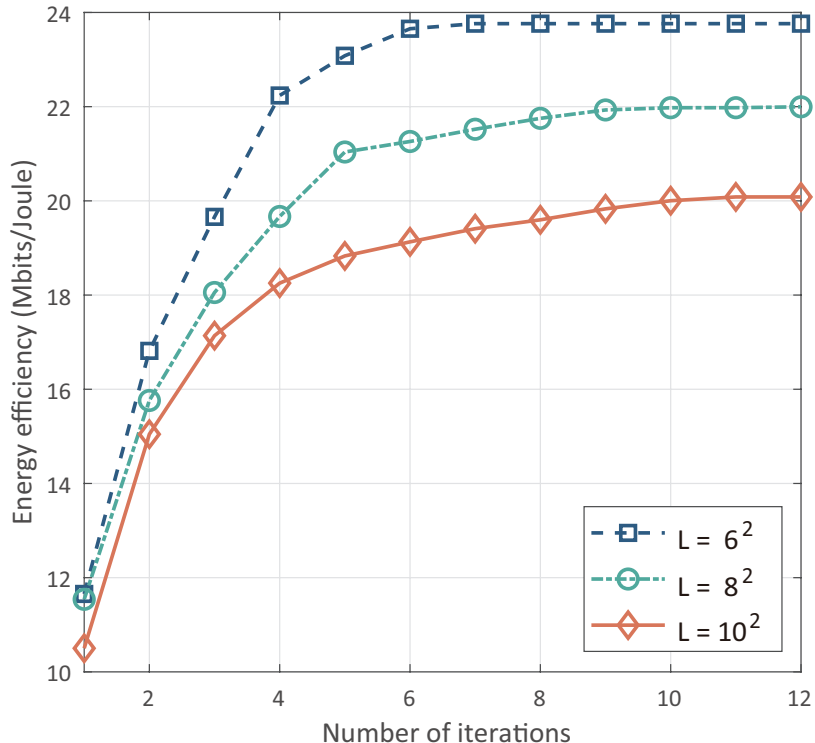}
\caption{Number of iterations to converge given different RIS sizes~($P_T = 30$~dBm).}
\label{Fig:iteration}
\end{figure}
\subsection{Convergence Speed}
Fig.~\ref{Fig:iteration} shows the energy efficiency vs. number of iterations for different sizes of each RIS.
The transmit power per BS $P_T$ is set as $30$~dBm and the number of quantization bits for discrete phase shifts~$b$ is set as $3$.
It can be observed that the convergence speed slows down when the size of each RIS $L$ grows since more variables need to be optimized.
We observe that the algorithm can converge within $12$ iterations for most cases.
Therefore, the convergence analysis provided in Section~\ref{sec:analysis}-A is verified, and the complexity is low enough to be acceptable.
\section{Conclusion \label{sec:conclusion}}%

In this paper, we considered an RIS aided cell-free MIMO system where several BSs are coordinated to serve various users with the assisting of multiple RISs.
To maximize the energy efficiency, we have proposed an HBF scheme where the digital beamforming and the RIS-based analog beamforming are performed at BSs and RISs, respectively.
The energy efficiency maximization problem has been formulated and solved by our proposed EEM algorithm in an iterative manner.
Simulation results show that, with the same transmit power in the entire system, the RIS aided cell-free system achieves better energy efficiency performance compared with the conventional ones.

Three remarks can be drawn from the theoretical analysis and numerical results, providing insights for the design of RIS aided cell-free systems.
\begin{itemize}
  \item The energy efficiency increases rapidly with the transmit power of each BS $P_T$ and gradually flattens when $P_T$ is large enough.
  \item There exists an optimal number of deployed RISs to maximize the energy efficiency.
  \item There exists a trade-off between energy efficiency and sum rate determined by the size of each RIS.
\end{itemize}

\begin{appendices}
\section{Definition of $a_1\sim a_8$ in (\ref{f_Q_de})}

The details of $a_1\sim a_8$ are defined in TABLE~\ref{tab:a}, where
$E_1 = \text{eig}(\mathbf{A}_j)$,
$E_2 = \text{Tr}(\mathbf{D}_j^{-2}\mathbf{C}_j)$,
$E_3 = \text{Tr}(\mathbf{D}_j^{-1}\mathbf{C}_j)$,
$E_4 = \text{Tr}(\mathbf{D}_j^{-2}\mathbf{C}_j^H)$,
$E_5 = \text{Tr}(\mathbf{D}_j^{-2}\mathbf{C}_j\mathbf{D}_j^{-1}\mathbf{C}_j^H+\mathbf{D}_j^{-1}\mathbf{C}_j\mathbf{D}_j^{-2}\mathbf{C}_j^H)$,
$E_6 = \text{Tr}[(\mathbf{D}_j^{-1}\mathbf{C}_j\mathbf{D}_j^{-1})^2\mathbf{C}_j^H]$,
$E_7 = \text{Tr}(\mathbf{D}_j^{-1}\mathbf{C}_j\mathbf{D}_j^{-1}\mathbf{C}_j^H)$,
and $E_8 = \text{Tr}(\mathbf{D}_j^{-1}\mathbf{C}_j^H)$.
\begin{table}[t]
\small
    \centering
    \caption{Definition of $a_1\sim a_8$}
    \begin{tabular}{c|c}
        \hline
        \hline
        \textbf{Parameter} & \textbf{Value}\\
        \hline
        \hline
        $a_1$   &   $E_1E_8-E_4$\\
        \hline
        $a_2$   &   $(2E_1E_3-E_2)E_8-E_1E_7+E_5-2E_3E_4+E_1$\\
        \hline
        $a_3$   &   $(E_1E_3^2-E_2E_3)E_8+(E_2-E_1E_3)E_7-E_6+E_3E_5-E_3^2E_4$\\
        \hline
        $a_4$   &   $E_1E_3^2-E_2E_3$\\
        \hline
        $a_5$   & $E_8$\\
        \hline
        $a_6$   &   $2E_3E_8-E_7+1$\\
        \hline
        $a_7$   &   $E_3^2E_8-E_3E_7+2E_3$\\
        \hline
        $a_8$   &   $E_3^2$\\
        \hline
        \hline
    \end{tabular}   \label{tab:a}
\end{table}
\section{Derivation of $\chi$ in (\ref{chi})}

Let $\frac{\partial f(\mathbf{Q})}{\partial \theta_{m,l}}=0$, through simplification, we have
\begin{equation}
  b_1e^{2j\theta_{m,l}}+b_2e^{j\theta_{m,l}}+b_3+b_4e^{-j\theta_{m,l}}+b_5e^{-2j\theta_{m,l}}=0,
\end{equation}
where $b_1 = j(a_1a_6-a_2a_5)$, $b_2 = 2j(a_1a_7-a_3a_5)$, $b_3 = j(3a_1a_8+a_2a_7-a_3a_6-3a_4a_5)$, $b_4 = 2j(a_2a_8-a_4a_6)$, and $b_5 = j(a_3a_8-a_4a_7)$.
According to the Euler's formula and the auxiliary Angle formula, this equation can be rewritten as
\begin{equation}\label{appB}
  \sqrt{\text{Re}\{b1+b5\}^2+\text{Im}\{b5-b1\}^2}\sin(2\theta_{m,l}+\gamma_1)+\sqrt{\text{Re}\{b2+b4\}^2+\text{Im}\{b4-b2\}^2}\sin(\theta_{m,l}+\gamma_2)=0,
\end{equation}
where $\tan\gamma_1=\frac{\text{Im}\{b5-b1\}}{\text{Re}\{b1+b5\}}$ and $\tan\gamma_2=\frac{\text{Im}\{b4-b2\}}{\text{Re}\{b2+b4\}}$.
After expanding the expression to the left of the equals by $\sin\alpha=\frac{2\tan\frac{\alpha}{2}}{1+\tan^2\alpha}$ and $\sin(\alpha+\beta)=\sin\alpha\cos\beta+\cos\alpha\sin\beta$, (\ref{appB}) can be rewritten as a quartic equation of one unknown, i.e., $\tan\frac{\theta_{m,l}}{2}$, which can be easily solved.
The solution of the equation is denoted by~$\chi$.
\section{Proof of Proposition~\ref{prop:M}} \label{appC}

Based on (\ref{powerpro}), the energy efficiency $\eta$ at \textbf{high SNR} can be rewritten as
\begin{equation}\label{anaEE}
  \eta \approx \frac{\sum_{k=1}^{K}\log_2(p_k\sigma^{-2})}{\omega\sum_{k=1}^{K}p_k + \mathcal{P}_s}
  =\frac{2BK\log_2M+B\sum_{k=1}^{K}log_2(L^2\mathbf{V}_{D,k}^H\mathbf{V}_{D,k}\sigma^{-2})}{\omega M^2\sum_{k=1}^{K} (L^2\mathbf{V}_{D,k}^H\mathbf{V}_{D,k})+M\cdot LP_R+(N_0-M)P_B+KP_U},
\end{equation}
where $N_0$ is the total number of BSs and RISs, i.e., $N_0 = N + M$.

Therefore, the derivative of energy efficiency $\eta$ with respect to the number of RISs $M$ can be given by
\begin{equation}\label{dM}
  \frac{\partial\eta}{\partial M}=\frac{c_1c_5+Mc_1c_4-M\ln Mc_1c_4-Mc_2c_4+M^2c_1c_3-2M^2\ln Mc_1c_3-2M^2c_2c_3}{Mc_5^2+(2M^2c_4+2M^3c_3)c_5+M^3c_4^2+2M^4c_3c_4+M^5c_3^2},
\end{equation}
where $c_1 = \frac{2BK}{\ln2}$, $c_2 = B\sum_{k=1}^{K}\log_2(L^2\mathbf{V}_{D,k}^H\mathbf{V}_{D,k}\sigma^{-2})$, $c_3 = \omega L^2\text{Tr}(\mathbf{V}_{D}^H\mathbf{V}_{D})$, $c_4=LP_R-P_B$, and $c_5=KP_U+N_0P_B$.
The denominator can be rewritten as $(M^{\frac{3}{2}}c_4+M^{\frac{1}{2}}c_5+M^{\frac{5}{2}}c_3)^2$ which is positive.

We now prove the molecule of (\ref{dM}) is negative.
Through simplification, the molecule can be rewritten as
\begin{equation}\label{moledM}
  \underbrace{(c_1-c_2)(Mc_1c_4+M^2c_3)}_{\textcircled{1}} + \underbrace{c_1c_5-M^2c_2c_3}_{\textcircled{2}} +\underbrace{(- M\ln M(c_1c_4+2Mc_1c_3))}_{\textcircled{3}},
\end{equation}
where the term $\textcircled{3}$ is obvious non-positive.
The valuence of term $\textcircled{1}$ and $\textcircled{2}$ are proved as follows.
\subsection{Proof of $\textcircled{1}<0$}\label{alessb}
The term $\textcircled{1}$ can be negative if $c_1<c_2$.
According to Jensen’s inequality~\cite{Jenson}, we have
\begin{equation}\label{alessbequ}
  \log_2(L^2\text{Tr}(\mathbf{V}_{D}^H\mathbf{V}_{D})\sigma^{-2})-\sum_{k=1}^{K}\log_2(L^2\mathbf{V}_{D,k}^H\mathbf{V}_{D,k}\sigma^{-2})\leq K\log_2K.
\end{equation}
Consider a high-SNR case, $K\log_2K$ is small enough to ignore compared to the terms in the left side which are equivalent to the sum rate in a cell-free system with $1$ RIS.
Hence, $c_2$ can be rewritten as
\begin{equation}
  c_2 = B\sum_{k=1}^{K}\log_2(L^2\mathbf{V}_{D,k}^H\mathbf{V}_{D,k}\sigma^{-2})
  \approx B\log_2(L^2\text{Tr}(\mathbf{V}_{D}^H\mathbf{V}_{D})\sigma^{-2}).
\end{equation}

Let $c_1<c_2$, we have
\begin{equation}
  \frac{2BK}{\ln2}<BK\log_2\left[\frac{L^2\text{Tr}(\mathbf{V}_{D}^H\mathbf{V}_{D})}{\sigma^2}\right]
\end{equation}

Through simplification, it can be given as
\begin{equation}
  \frac{L^2\text{Tr}(\mathbf{V}_{D}^H\mathbf{V}_{D})}{\sigma^2}>e^2
\end{equation}
where $e$ is Euler number, $\text{Tr}(\mathbf{V}_{D}^H\mathbf{V}_{D})\sigma^{-2}$ can be considered as a value greater than $1$ when the SNR is high.
Since each RIS relies on the combination of multiple programmable radiating elements to realize signal reflection, in general, the number of elements of each RIS $L$ is larger than $e$.
Therefore, $c_1<c_2$ holds, and thus, term $\textcircled{1}$ is negative.
\subsection{Proof of $\textcircled{2}<0$}
Based on the assumption in Proposition~\ref{prop:M}, i.e., $M\times L\rightarrow\infty$, we have
\begin{equation}
  KP_U+NP_B<\omega M^2L^2\text{Tr}(\mathbf{V}_{D}^H\mathbf{V}_{D}).
\end{equation}
Specifically, $c_5<M^2c_3$ holds.
Multiply both sides by $c_1$, we have $c_1c_5<M^2c_1c_3$.
Based on the result in Appendix~\ref{alessb}, $M^2c_1c_3<M^2c_2c_3$, and thus, $c_1c_5<M^2c_2c_3$.
Therefore, term $\textcircled{2}$ is negative.
\end{appendices}
\vspace{-1.5em}
\section{Proof of Proposition~\ref{prop:size}}\label{app:size}
When the size of each RIS tends to 1, i.e., $L\rightarrow 1$, the SNR at user $k$ can be expressed by
\begin{equation}
\begin{aligned}
  SNR_k & = \frac{|\mathbf{H}_k\mathbf{V}_{D,k}|^2}{\sum_{k=1}^{K}|\mathbf{H}_k\mathbf{V}_{D,k}|^2+\sigma^2}\\
  &\overset{(a)}=|\mathbf{H}_k\mathbf{V}_{D,k}|^2/\sigma^2\\
  &=\sum_{m=1}^{M}\sum_{l=1}^{L}\mathbf{h}_{m,l,k}\mathbf{h}_{m,l,k}^H\mathbf{V}_{D,k}^H\mathbf{V}_{D,k}/\sigma^2\\
  &=ML\mathbb{E}[\mathbf{h}_{m,l,k}\mathbf{h}_{m,l,k}^H]\mathbf{V}_{D,k}^H\mathbf{V}_{D,k}/\sigma^2\\
  &=MLa_k,
  \end{aligned}
\end{equation}
where $a_k=\mathbb{E}[\mathbf{h}_{m,l,k}\mathbf{h}_{m,l,k}^H]\mathbf{V}_{D,k}^H\mathbf{V}_{D,k}/\sigma^2$ and $\mathbf{h}_{m,l,k}$ is the channel between the BSs to user $k$ via the $l$-th element of RIS $m$.
$(a)$ is obtained by the properties of ZF beamformer as given in~(\ref{ZF-proper}).

According to the property of the logarithmic function~\cite{liangletter}, we have $\mathbb{E}[\log_2(1+SNR_k)]\approx \log_2(1+\mathbb{E}[SNR_k])$, the energy efficiency can be rewritten as
\begin{equation}
\begin{aligned}
\eta& = \frac{BR}{\mathcal{P}}\\
&=\frac{B\sum_{k=1}^{K}\log_2(1+MLa_k)}{\omega\sum_{k=1}^{K}p_k+MLP_R+kP_U+NP_B}\\
&\approx \frac{B\log_2(1+\frac{\sum_{k=1}^{K}a_k}{K})}{\omega\sum_{k=1}^{K}p_k+MLP_R+kP_U+NP_B}.
\end{aligned}
\end{equation}
Therefore, the derivative of energy efficiency $\eta$ with respect to the size of each RIS $L$ can be given by
\begin{equation}\label{dL}
\frac{\partial\eta}{\partial L}=\frac{\frac{B}{\ln2}(\omega\sum_{k=1}^{K}p_k+KP_U+NP_B)+[(L-L\ln L)\frac{B}{\ln2}-LB\log_2(\frac{M}{K}\omega\sum_{k=1}^{K}a_k)]MP_R}{L(\omega\sum_{k=1}^{K}p_k+KP_U+NP_B)^2+2L^2MP_R(\omega\sum_{k=1}^{K}p_k+KP_U+NP_B)+L^3M^2P_R^2}.
\end{equation}

Obviously, the denominator is positive.
We define the molecule as $g(L)$, let $g(L)>0$, we have
\begin{equation}
  \omega\sum_{k=1}^{K}p_k+\mathcal{P}_s > \ln\left(\frac{M}{K}\sum_{k=1}^{K}a_k\right)MP_R
\end{equation}
where the left side is not less than $\mathcal{P}_s$ and the right side is not larger than $\ln\left(\frac{NP_T}{\sigma^2}\right)\cdot MP_R$.
Therefore, when the RIS aided cell-free system satisfies that $\frac{\mathcal{P}_s}{M\cdot P_R}>\ln{\frac{NP_T}{k\sigma^2}}$, we have $g(L)>0$.

Based on Lemma~\ref{lemmaL}, the energy efficiency of RIS aided cell-free system increases when the size of each RIS $L$ is small, and then gradually drops to zero if $\frac{\mathcal{P}_s}{M\cdot P_R}>\ln{\frac{NP_T}{k\sigma^2}}$.

This completes the proof.

\end{document}